\pdfoutput=1
\newif\ifFull
\Fullfalse
\documentclass[11pt]{article}
\advance \topmargin by -\headheight
\advance \topmargin by -\headsep
\evensidemargin \oddsidemargin
\usepackage{amsfonts}
\usepackage{amsmath}
\usepackage{graphicx}
\usepackage{color}
\usepackage{times}
\usepackage{subfigure}
\usepackage[noend]{algorithmic}
\usepackage{url}

\newcommand{\newparentheses}[3]{%
  \expandafter\newcommand\csname #1\endcsname[1]{#2##1#3}%
  \expandafter\newcommand\csname #1L\endcsname[1]{\bigl#2##1\bigr#3}%
  \expandafter\newcommand\csname #1XL\endcsname[1]{\Bigl#2##1\Bigr#3}%
  \expandafter\newcommand\csname #1V\endcsname[1]{\left#2##1\right#3}}

\newparentheses{parens}{(}{)}
\newparentheses{braces}{\{}{\}}
\newparentheses{brackets}{[}{]}
\newparentheses{floor}{\lfloor}{\rfloor}
\newparentheses{ceil}{\lceil}{\rceil}
\newparentheses{abs}{|}{|}
\newparentheses{set}{\{}{\}}
\newparentheses{size}{|}{|}
\newparentheses{seq}{\langle}{\rangle}
\newparentheses{pair}{\langle}{\rangle}

\makeatletter
\newcommand{\onenewattribute}[4]{%
  \@ifundefined{#2}{\let\@@def\newcommand}{\let\@@def\renewcommand}%
  \expandafter\@@def\csname #2\endcsname[1][]{%
    \def\first@arg{##1}\csname @#2\endcsname}%
  \@ifundefined{@#2}{\let\@@def\newcommand}{\let\@@def\renewcommand}%
  \expandafter\@@def\csname @#2\endcsname[2][]{%
    \ifthenelse{\equal{#1}{sub}}%
    {\csname @@#2\endcsname{##1}{\first@arg}{##2}}%
    {\csname @@#2\endcsname{\first@arg}{##1}{##2}}}
  \@ifundefined{@@#2}{\let\@@def\newcommand}{\let\@@def\renewcommand}%
  \expandafter\@@def\csname @@#2\endcsname[3]{%
    \ifthenelse{\equal{##1}{}}%
    {\ifthenelse{\equal{##2}{}}%
      {#3\csname #4\endcsname{##3}}%
      {#3_{##2}\csname #4\endcsname{##3}}}%
    {\ifthenelse{\equal{##2}{}}%
      {#3^{##1}\csname #4\endcsname{##3}}%
      {#3_{##2}^{##1}\csname #4\endcsname{##3}}}}}
\newcommand{\newattribute}[3][sub]{%
  \onenewattribute{#1}{#2}{#3}{parens}%
  \onenewattribute{#1}{#2L}{#3}{parensL}%
  \onenewattribute{#1}{#2XL}{#3}{parensXL}%
  \onenewattribute{#1}{#2V}{#3}{parensV}}

\newcommand{\newproperty}[3][sub]{%
  \@ifundefined{#2}{\let\@@def\newcommand}{\let\@@def\renewcommand}%
  \expandafter\@@def\csname #2\endcsname[2][]{%
    \ifthenelse{\equal{#1}{sub}}%
    {\ifthenelse{\equal{##1}{}}%
      {#3_{##2}}%
      {#3_{##2}^{##1}}}%
    {\ifthenelse{\equal{##1}{}}%
      {#3^{##2}}%
      {#3_{##1}^{##2}}}}}
\makeatother

\newattribute{OhOf}{{\cal O}}
\newattribute{ThetaOf}{\Theta}
\newattribute{OmegaOf}{\Omega}
\newattribute{ohOf}{\mathrm{o}}
\newattribute{omegaOf}{\omega}

\newattribute{scc}{\textrm{scc}}
\newattribute{pred}{\textrm{pred}}

\def\comment#1{}
\def\withcomments{%
  \addtolength{\oddsidemargin}{-0.5in}%
  \addtolength{\evensidemargin}{-0.5in}%
  \setlength{\marginparwidth}{1in}
  \newcounter{mycommentcounter}%
  \def\comment##1{\refstepcounter{mycommentcounter}%
    \ifhmode
      \unskip
      {\dimen1=\baselineskip
        \divide\dimen1 by 2%
        \raise\dimen1\llap{\tiny -\themycommentcounter-}}%
    \fi
    \marginpar{\renewcommand{\baselinestretch}{0.8}%
      \footnotesize [\themycommentcounter]: \raggedright ##1}}%
  \date{\framebox{Draft of \today}}}

\newcommand{\OUT}{{\cal OUT}}
\newcommand{\IN}{{\cal IN}}
\renewcommand{\emph}[1]{\textit{\textbf{#1}}}

\makeatletter
\def\@begintheorem#1#2{\sl \trivlist \item[\hskip \labelsep{\bf #1\
#2:}]}
\def\@opargbegintheorem#1#2#3{\sl \trivlist
      \item[\hskip \labelsep{\bf #1\ #2\ #3:}]}
\makeatother
\newenvironment{proof}{\noindent{\bf Proof:}}{\hspace*{\fill}\rule{6pt}{6pt}\bigskip}
\newtheorem{theorem}{Theorem}[section]
\newtheorem{lemma}[theorem]{Lemma}

\newcommand{\cT}{{\cal T}}

\usepackage{amssymb}

\begin{document}

\title{
Sorting, Searching, and Simulation in the MapReduce Framework}

\author{%
Michael T. Goodrich \\
goodrich@ics.uci.edu \\
Center for Algorithms and Theory of Computation \\
Department of Computer Science \\
University of California, Irvine \\
Irvine, CA 92697-3435, USA
\and
\begin{tabular}{c c}
Nodari Sitchinava & Qin Zhang\\
nodari@madalgo.au.dk  & qinzhang@madalgo.au.dk  
\end{tabular}
\\
MADALGO -- Center for Massive Data Algorithmics \\
Department of Computer Science \\
Aarhus University \\
IT-Parken, Aabogade 34 \\
DK-8200, Aarhus N, DENMARK 
}

\date{}

\maketitle 

\begin{abstract} 
In this paper, we study the
MapReduce framework from an
algorithmic standpoint and demonstrate the usefulness of our approach
by designing and analyzing efficient 
MapReduce algorithms for fundamental sorting, searching, and
simulation problems.
This study is motivated by a goal of
ultimately putting the MapReduce framework on an equal theoretical
footing with the well-known PRAM and BSP parallel models, which would
benefit both the theory and practice of MapReduce algorithms.
We describe efficient MapReduce algorithms for sorting,
   multi-searching, and simulations of parallel algorithms specified in the BSP
   and CRCW PRAM models.  We also provide some applications of these results to
   problems in parallel computational geometry for the MapReduce framework,
   which result in efficient MapReduce algorithms for sorting, 2- and
   3-dimensional convex hulls, and fixed-dimensional linear programming.  For
   the case when mappers and reducers have a
     memory/message-I/O size of $M=\ThetaOf{N^\epsilon}$, for a small constant
     $\epsilon>0$, all of our MapReduce algorithms for these applications run in
     a constant number of rounds.
\end{abstract}


\ifFull\else
\setcounter{page}{0}
\thispagestyle{empty}
\fi

\clearpage
\section{Introduction}
\label{sec:intro}

The {\em MapReduce framework}~\cite{dg-msdpl-08,dg-mafdp-10} is a
programming paradigm for designing parallel and distributed algorithms. It
provides a simple programming interface that is specifically designed to
make it easy for a programmer to design a parallel program that can
efficiently perform a data-intensive computation. Moreover, it is a
framework that allows for parallel programs to be directly translated into
computations for cloud computing environments and server clusters (e.g.,
see~\cite{rrbn-smlsa-09}). 
This framework is gaining wide-spread interest
in systems domains, in that this framework is being used in Google data
centers and as a part of the open-source Hadoop system~\cite{w-hdg-09} for
server clusters, which have been deployed by a wide variety of
enterprises\footnote{See \url{http://en.wikipedia.org/wiki/Hadoop}.},
including Yahoo!, IBM, The New York Times, eHarmony, Facebook, and Twitter.

Building on pioneering work by Feldman {\it
  et al.}~\cite{DBLP:conf/soda/FeldmanMSSS08} and Karloff {\it et
  al.}~\cite{ksv-amcfm-10}, our interest in this paper is in studying the
MapReduce framework from an
algorithmic standpoint, by designing and
analyzing MapReduce algorithms for fundamental sorting, searching, and
simulation problems. Such a study could be a step on the way
to ultimately putting the MapReduce framework on an equal theoretical
footing with the well-known PRAM and BSP parallel models.

Still, we would be remiss if we did not mention that this framework is not
without its detractors. DeWitt and Stonebraker~\cite{ds-mmsb-08} mention
several issues they feel are shortcomings of the MapReduce framework,
including that it seems to require brute-force enumeration instead of
indexing for performing searches. Naturally, we feel that this criticism is
a bit harsh, as the theoretical limits of the MapReduce framework have yet
to be fully explored; hence, we feel that further theoretical study is
warranted. Indeed, this paper can be viewed as at least a partial
refutation of the claim that the MapReduce framework disallows indexed
searching, in that we show how to perform fast and efficient multi-search
in the MapReduce framework.

\subsection{The MapReduce Framework}
In the MapReduce framework, a computation is specified as a sequence of
map, shuffle, and reduce steps that operate on a set
$X=\{x_1,x_2,\ldots,x_n\}$ of values:
\begin{itemize}
\item
A {\em map step} applies a function, $\mu$, to each value, $x_i$, to
produce a finite set of key-value pairs $(k,v)$. To allow for parallel
execution, the computation of the function $\mu(x_i)$ must depend only on
$x_i$.

\item 
A {\em shuffle step} collects all the key-value pairs produced in the
previous map step, and produces a set of lists, $L_k=(k;v_1,v_2,\ldots)$,
where each such list consists of all the values, $v_j$, such that $k_j=k$
for a key $k$ assigned in the map step.

\item
A {\em reduce step} applies a function, $\rho$, to each list
$L_k=(k;v_1,v_2,\ldots)$, formed in the shuffle step, to produce a set of
values, $y_1,y_2,\ldots\,$. The reduction function, $\rho$, is allowed to
be defined sequentially on $L_k$, but should be independent of other lists
$L_{k'}$ where $k'\not=k$.
\end{itemize}

The parallelism of the MapReduce framework comes from the fact that each
map or reduce operation can be executed on a separate processor
independently of others. Thus, the user simply defines the functions $\mu$
and $\rho$, and the system automatically schedules map-shuffle-reduce steps
and routes data to available processors, including provisions for fault
tolerance.

The outputs from a reduce step can, in general, be used as inputs to
another round of map-shuffle-reduce 
steps.
Thus, a typical MapReduce computation is described as a sequence of
map-shuffle-reduce steps that perform a desired action in a series of {\em
  rounds} that produce the algorithm's output after the last reduce step.

\subsection{Evaluating MapReduce Algorithms}
Ideally, we desire the number of rounds in a MapReduce algorithm to be a
constant. For example, consider an often-cited MapReduce algorithm to count
all the instances of words in a document. Given a document, $D$, we define
the set of input values $X$ to be all the words in the document and we then
proceed as follows:
\begin{enumerate}
\item
Map: For each word, $w$, in the document, map $w$ to $(w, 1)$.
\item
Shuffle: collect all the $(w, 1)$ pairs for each word, producing a list
$(w;1,1,\ldots,1)$, noting that the number of $1$'s in each such list is
equal to the number of times $w$ appears in the document.
\item
Reduce: scan each list $(w;1,1,\ldots,1)$, summing up the number of $1$'s
in each such list, and output a pair $(w,n_w)$ as a final output value,
where $n_w$ is the number of $1$'s in the list for $w$.
\end{enumerate}
This single-round computation clearly computes the number of times
each word appears in $D$.

The number of rounds in a MapReduce algorithm is not always equal to
$1$, however, and there are, in fact, several metrics that one can use to
measure the efficiency of a MapReduce algorithm over the course of its
execution, including the following:
\begin{itemize}
\item
We can consider $R$, the {\em number of rounds} of map-shuffle-reduce that
the algorithm uses.
\item
If we let $n_{r,1},n_{r,2},\ldots$ denote the mapper and reducer I/O sizes
for round $r$, so that $n_{r,i}$ is the size of the inputs and outputs for
mapper/reducer $i$ in round $r$, then we can define $C_r$, the {\em
  communication complexity of round $r$}, to be the total size of the
inputs and outputs for all the mappers and reducers in round $r$, that is,
$C_r=\sum_i n_{r,i}$. We can also define the {\em communication
  complexity}, $C=\sum_{r=0}^{R-1}C_r$, for the entire algorithm.
\item
We can let $t_r$ denote the {\em internal running time} for round $r$,
which is the maximum internal running time taken by a mapper or reducer in
round~$r$, where we assume $t_r\ge \max_i \{n_{r,i}\}$, since a mapper or
reducer must have a running time that is at least the size of its inputs
and outputs. We can also define {\em total internal running time},
$t=\sum_{r=0}^{R-1} t_r$, for the entire algorithm, as well.
\end{itemize}
We can make a crude calibration of a MapReduce algorithm using the
following additional parameters:
\begin{itemize}
\item
$L$: the latency $L$ of the shuffle network, which is the number of steps
  that a mapper or reducer has to wait until it receives its first input in
  a given round.
\item
$B$: the bandwidth of the shuffle network, which is the number of elements
  in a MapReduce computation that can be delivered by the\ shuffle network
  in any time unit.
\end{itemize}

Given these parameters, a lower bound for the total running time, $T$, of
an implementation of a MapReduce algorithm can be characterized as
follows: $$ T = \Omega\left(\sum_{r=0}^{R-1} (t_r + L + C_r/B)\, \right) =
\Omega(t + RL + C/B).$$ For example, given a document $D$ of $n$ words, the
simple word-counting MapReduce algorithm given above has a worst-case
performance of $R = 1$, $C = \ThetaOf{n}$, and $t = \ThetaOf{n}$; hence,
its worst-case time performance $T = \ThetaOf{n}$, which is no faster than
sequential computation.  
Unfortunately, such performance could be quite
common, as the frequency of words in a natural-language document tend to
follow Zipf's law, so that some words appear quite frequently, and the
running time of the simple word-counting algorithm is proportional to the
number of occurrences of the most-frequent word. For instance, in the Brown
Corpus~\cite{kf-capda-67}, the word ``the'' accounts for 7\% of all word
occurrences.\footnote{\url{http://en.wikipedia.org/wiki/Zipf's_law}}

Note, therefore, that focusing exclusively on $R$, the number of rounds in
a MapReduce algorithm, can actually lead to an inefficient algorithm. For
example, if we focus only on the number of rounds, $R$, then the most
efficient algorithm would always be the {\em trivial one-round algorithm},
which maps all the inputs to a single key and then has the reducer for this
key perform a standard sequential algorithm to solve the problem. This
approach would run in one round, but it would not use any parallelism;
hence, it would be relatively slow compared to an algorithm that was more
``parallel.''

\subsection{Memory-Bound and I/O-Bound MapReduce Algorithms}
So as to steer algorithm designers away from the trivial one-round
algorithm, recent algorithmic formalizations of the MapReduce paradigm have
focused primarily on optimizing the round complexity bound, $R$, while
restricting the memory size or input/output size for reducers. Karloff {\it
  et al.}~\cite{ksv-amcfm-10} define their MapReduce model, MRC, so that
each reducer's I/O size is restricted to be $\OhOf{n^{1-\epsilon}}$ for some
small constant $\epsilon>0$, and Feldman {\it et
  al.}~\cite{DBLP:conf/soda/FeldmanMSSS08} define their model, MUD, so that
reducer memory size is restricted to be $\OhOf{\log^c n}$, for some constant
$c\ge 0$, and reducers are further required to process their inputs in a
single pass. These restrictions limit the feasibility of the trivial
one-round algorithm for solving a problem in the MapReduce framework and
instead compel algorithm designers to make better utilization of
parallelism.

In this paper, we follow the I/O-bound approach, as it seems to correspond
better to the way reducer computations are specified, but we take a
somewhat more general characterization than Karloff {\it et
  al.}~\cite{ksv-amcfm-10}, in that we do not bound the I/O size for
reducers explicitly to be $\OhOf{n^{1-\epsilon}}$, but instead allow it to be
an arbitrary parameter:
\begin{itemize}
\item
We define $M$ to be an upper bound on the {\em I/O-buffer memory size} for
all reducers used in a given MapReduce algorithm. That is, we predefine $M$
to be a parameter and require that $ \forall r, i:\ n_{r,i} \le M.  $
\end{itemize}
We then can use $M$ in the design and/or analysis of each of our MapReduce
algorithms. For instance, if each round of an algorithm has a reducer that
with an I/O size of at most $M$, then we say that this algorithm is an {\em
  I/O-memory-bound MapReduce algorithm} with parameter $M$. In addition, if
each round has a reducer with an I/O size proportional to $M$ (whose
processing probably dominates the reducer's internal running time), then we
can give a simplified lower bound on the time, $T$, for such an algorithm
as
\[
T = \Omega(R(M+L) + C/B).
\]

This approach therefore can characterize the limits of parallelism that are
possible in a MapReduce algorithm and it also shows that we should
concentrate on the round complexity and communication complexity of a
MapReduce algorithm in characterizing its performance\footnote{These
  measures correspond naturally with the {\em time} and {\em work} bounds
  used to characterize PRAM algorithms (e.g., see~\cite{j-ipa-92}).}.  Of
course, such bounds for $R$ and $C$ may depend on $M$, but that is fine,
for similar characterizations are common in the literature on
external-memory algorithms (e.g.,
see~\cite{av-iocsr-88,a-emaa-97,a-emds-02,
  v-emads-01,DBLP:reference/algo/Vitter08}). In the rest of the paper, when
we talk about the MapReduce model, we always mean the I/O-memory-bound
MapReduce model except when mentioned explicitly.

\subsection{Our Contributions}
We provide several efficient algorithms in the MapReduce framework for
fundamental combinatorial problems, including parallel prefix-sum,
multi-search, and sorting. All of these algorithms run in $\OhOf{\log_M N}$
map-shuffle-reduce rounds with high probability; hence, they are constant-round computations for
the case when $M$ is $\Theta(N^\epsilon)$ for some constant
$\epsilon>0$. 

Unlike the sorting algorithm in the original paper describing MapReduce
framework~\cite{dg-msdpl-08}, our sorting algorithm is truly parallel for
it does not require a central master node to compute partitioning
sequentially.

What is perhaps most unusual about the MapReduce framework is that there is
no explicit notion of ``place'' for where data is stored nor for where
computations are performed. 
This property of the MapReduce framework is perhaps what led
DeWitt and Stonebraker~\cite{ds-mmsb-08} to say that it does not
support indexed searches.
Nevertheless, we show that the MapReduce framework does in fact
support efficient multi-searching, where one is interested in
searching for a large number of keys in a search tree of roughly
equal size.

We also provide a number of simulation results. We show that any
Bulk-Synchronous Parallel (BSP) algorithm~\cite{v-bmpc-90} running in $R$
super-steps with a memory of size $N$ and $P\le N$ processors can be
simulated with a MapReduce algorithm in $R$ rounds and communication
complexity $C = \OhOf{RN}$ with reducer I/O-buffers of size $M =
\OhOf{N/P}$. We also show that any CRCW PRAM algorithm running in $T$ steps
with $P$ processors on a memory of size $N$ can be simulated in the
MapReduce framework in $R = \OhOf{T\log_M P}$ rounds with $C =
\OhOf{T(N+P)\log_M (N+P)}$ communication complexity. This latter simulation
result holds for any version of the CRCW PRAM model, including the $f$-CRCW
PRAM, which involves the computation of a commutative semigroup operator
$f$ on concurrent writes to the same memory location, such as in the
Sum-CRCW PRAM~\cite{e-oascp-07}.  The PRAM simulation results achieve their
efficiency through the use of a technique we call the {\em invisible
  funnel} method, as it can be viewed as placing virtual multi-way trees
rooted at the input items. These trees funnel concurrent read and write
requests to the data items, but are never explicitly constructed.  The
simulation results can be applied to solve several parallel computational
geometry problems, including convex hulls 
and fixed-dimensional linear programming.

\paragraph{Roadmap.} The rest of the paper is organized as follows. In
Section~\ref{sec:genericMR}, we first present our generic MapReduce
framework which simplifies the development and exposition of algorithms
that follow.  In Section~\ref{sec:simulation}, we show how to simulate BSP
and CRCW PRAM algorithms in the MapReduce framework.  Finally in
Section~\ref{sec:search-sorting}, we design MapReduce algorithms for
multi-search and sorting.

\section{Generic MapReduce Computations} \label{sec:genericMR}
In this section we define an abstract computational model that captures the
MapReduce framework.

Consider a set of nodes $V$. Let $A_v(r)$ be a set of items associated with
each node $v \in V$ at round $r$, which defines the state of $v$. Also, let
$f$ be a sequential function defined for all nodes. Function $f$ takes as input the state
$A_v(r)$ of a node $v$ and returns a new set $B_v(r)$, in the process
destroying $A_v(r)$. Each item of $B_v(r)$ is of the form $(w, a)$,
where $w \in V$ and $a$ is a new item.  We define the following computation
which proceeds in $R$ rounds.

At the beginning of the computation only the {\em input nodes} $v$ have
non-empty states $A_v(0)$. The state of an input node consists of a single
input item.

In round $r$, each node $v$ with non-empty state $A_v(r) \not = \emptyset$
performs the following. First, $v$ applies function $f$ on $A_v(r)$. This
results in the new set $B_v(r)$ and deletion of $A_v(r)$. Then, for each
element $b = (w, a) \in B_v(r)$, node $v$ sends item $a$ to node $w$. Note
that if $w = v$, then $v$ sends $a$ back to itself.  As a result of this
process, each node may receive a set of items from others. Finally, the set
of received items at each node $v$ defines the new state $A_v(r+1)$ for the
next round. The items comprising the non-empty states $A_v(r)$ after $R$
rounds define the outputs of the entire computation at which point the
computation halts.

The number of rounds $R$ denotes the {\em round complexity} of the
computation.  The total number of all the items sent (or, equivalently,
received) by the nodes in each round $r$ defines the {\em communication
  complexity} $C_r$ of round $r$, that is, $C_r = \sum_v
|B_v(r)|$. Finally, the communication complexity $C$ of the entire
computation is defined as $C = \sum_{r=0}^{R-1} C_r = \sum_{r=0}^{R-1}
\sum_v |B_v(r)|$. Note that this definition implies that nodes $v$ whose
states $A_v(r)$ are empty at the beginning of round $r$ do not contribute
to the communication complexity. Thus, the set $V$ of nodes can be
infinite. But, as long as only a finite number of nodes have non-empty
$A_v(r)$ at the beginning of each round, the communication complexity of
the computation is bounded.

Observe that during the computation, in order for node $v$ to send items to
node $w$ in round $r$, $v$ should know the label of the destination $w$,
which can be obtained by $v$ in the following possible ways (or any
combination thereof): 1) the link $(v, w)$ can be encoded in $f$ as a
function of the label of $v$ and round $r$, 2) some node might send the
label of $w$ to $v$ in the previous round, or 3) node $v$ might keep the
label of $w$ as part of its state by constantly sending it to itself.

Thus, the above computation can be viewed as a computation on a dynamic
directed graph $G = (V, E)$, where an edge $(v,w) \in E$ in round $r$
represents a possible communication link between $v$ and $w$ during that
round. The encoding of edges $(v,w)$ as part of function $f$ is equivalent
to defining an {\em implicit} graph~\cite{kannan:implicitgraphs}; keeping
all edges within a node throughout the computation is equivalent to
defining a {\em static} graph.
For ease of exposition, we define the following primitive operations that
can be used within $f$ at each node $v$:

\begin{itemize}
\item create an item; delete an item; modify an item; keep item $x$ (that
  is, the item $x$ will be sent to $v$ itself by creating an item $(v, x)
  \in B_v(r)$); send an item $x$ to node $w$ (create an item $(w,x) \in
  B_v(r)$).

\item create an edge; delete an edge. This is essentially the same as
  create an item and delete an item, since explicit edges are just
  maintained as items at nodes. This operations will simplify exposition
  when dealing with explicitly defined graphs $G$ on which computation is
  performed.
\end{itemize}

The following theorem shows that the above framework captures the essence
of computation in the MapReduce framework:

\begin{theorem} \label{thm:framework}
Let $G = (V,E)$ and $f$ be defined as above such that in each round each
node $v \in V$ sends, keeps and receives at most $M$ items. Then computation on $G$ with
round complexity $R$ and communication complexity $C$ can be simulated in
the I/O-memory-bound MapReduce model with the same round and
communication complexities.
\end{theorem}

\begin{proof}
We implement round $r = 0$ of computation on $G$ in the I/O-memory-bound
MapReduce framework using only the Map and Shuffle steps and every round
$r>0$ using the Reduce step of round $r-1$ and a Map and Shuffle step of
round $r$.
\begin{enumerate}
\item Round $r = 0$: (a) Computing $B_v(r) = f(A_v(r))$: Initially, only
  the input nodes have non-empty sets $A_v(r)$, each of which contains only
  a single item.  Thus, the output $B_v(r)$ only depends on a single item,
  fulfilling the requirement of Map. We define Map to be the same as $f$, i.e.,
  it outputs a set of key-value tuples $(w, x)$, each of which corresponds to
  an item $(w, x)$ in $B_v(r)$.  (b) Sending items to destinations: The
  Shuffle step on the output of the Map step ensures that all tuples with
  key $w$ will be sent to the same reducer, which corresponds to the node
  $w$ in $G$.

\item Round $r > 0$: First, each reducer $v$ that receives a tuple $(v;
  x_1, x_2, \dots, x_k)$ (as a result of the Shuffle step of the previous
  round) simulates the computation at node $v$ in $G$. That is, it
  simulates the function $f$ and outputs a set of tuples $(w, x)$, each
  of which corresponds to an item in $B_v(r)$. We then define Map to be the
  identity map: On input $(w, x)$, output key-value pair $(w, x)$.
  Finally, the Shuffle step of round $r$ completes the simulation of the
  round $r$ of computation on graph $G$ by sending all tuples with key $w$
  to the same reducer that will simulate node $w$ in $G$ in round $r+1$.
\end{enumerate}
Keeping an item is equivalent to sending it to itself, thus, each node in
$G$ sends and receives at most $M$ items. Therefore, no reducer
receives or generates more than $M$ items implying that the above is
a correct I/O-memory-bound MapReduce algorithm.
\end{proof}

The above theorem gives an abstract way of designing MapReduce
algorithms. More precisely, to design a MapReduce algorithm, we define
graph $G$ and a sequential function $f$ to be performed at each node $v \in
V$. This is akin to designing BSP algorithms and is more intuitive way than
defining Map and Reduce functions. 

Note that in the above framework we can easily implement a global loop
primitive spanning over multiple rounds: Each item maintains a counter that
is updated at each round. We can also implement {\em parallel tail recursion} by
defining the labels of nodes to include the recursive call stack
identifiers.

Next, we show how we can implement an all-prefix-sum algorithm in the
generic MapReduce model. This algorithm will then be used as a subroutine
in a random indexing algorithm, which in turn will be used in the
multi-search algorithm in Section~\ref{sec:search}.

\subsection{Prefix Sums and Random Indexing} \label{sec:prefix}
The all-prefix-sum problem is usually defined on an array of
integers. Since there is no notion of arrays in the MapReduce framework,
but rather a collection of items, we define the all-prefix-sum problem as
follows: given a collection of items $x_i$, where $x_i$ holds an integer
$a_i$ and an index value $0 \le i \le N-1$, compute for each item $x_i$ a
new value $b_i = \sum_{j=0}^i a_j$.

The MapReduce algorithm for all-prefix-sum problem is the following. Graph
$G = (V,E)$ is an undirected\footnote{Each undirected edge is represented
  by two directed edges in $G$.}  rooted tree $\cT$ with branching factor
$d = M/2$ and height $L = \ceil{\log_d N} = \OhOf{\log_M N}$. The
root of the tree is defined to be at level $0$ and leaves at level
$L-1$. We label the nodes in $\cT$ such that the $k$-th node (counting from
the left) on level $l$ is defined as $v = (l, k)$. Then, we can identify
the parent of a non-root node $v = (l, k)$ as $p(v) = (l-1, \floor{k/d})$
and the $j$-th child of $v$ as $w_j = (l+1, k\cdot d+j)$.  In other words,
the neighborhood set of any node $v \in \cT$ can be computed solely from
the label of $v$, thus, we do not have to maintain edges explicitly.

In the initialization step, each input node simply sends its input item
$a_i$ with index $i$ to the leaf node $v = (L-1, i)$. The rest of the
algorithm proceeds in two phases, processing the nodes in $\cT$ one
level at a time. The nodes at other levels simply keep the items they have
received during previous rounds.

\begin{enumerate}

\item{Bottom-up phase.} For $l = L-1$ downto $1$ do: For each node $v$ on
  level $l$ do: If $v$ is a leaf node, it received a single value $a_i$
  from an input node.  The function $f$ at $v$ creates a copy $s_v = a_i$,
  keeps $a_i$ it had received and sends $s_v$ to the parent $p(v)$ of
  $v$. If $v$ is a non-leaf node, let $w_0, w_1, \dots, w_{d-1}$ denote
  $v$'s child nodes in the left-to-right order. Node $v$ received a set of
  $d$ items $A_v(r) = \{s_{w_0}, s_{w_1}, \dots, s_{w_{d-1}}\}$ from its
  children at the end of the previous round.  $f(A_v(r))$ computes the sum
  $s_v = \sum_{j=0}^{d-1} s_{w_j}$, sends $s_v$ to $p(v)$ and keeps all the
  items received by the children. 

\item{Top-down phase.} For $l = 0$ to $L-1$ do: For each node $v$ on level
  $l$ do: If $v$ is the root, it had received items $A_v(r) = \{s_{w_0},
  s_{w_1}, \dots, s_{w_{d-1}}\}$ at the end of the bottom-up phase. It
  creates for each child $w_i\ (0 \le i \le d-1)$ a new item $s'_i =
  \sum_{j=0}^{i-1} s_{w_j}$ and sends it to $w_i$. If $v$ is a non-root
  node, let $s_{p(v)}$ be the item received from its parent in the previous
  round. Inductively, the value $s_{p(v)}$ is the sum of all items ``to the
  left" of $v$. If $v$ is a leaf having a unique item $a_k$, then it simply
  outputs $a_k + s_{p(v)}$ as a final value, which is the prefix sum
  $\sum_{j=0}^k a_j$. Otherwise, it creates for each child $w_i\ (0 \le i
  \le d-1)$ a new item $s_{p(v)} + \sum_{j=0}^{i-1} s_{w_j}$ and sends it
  to $w_i$. In all cases, all items of $v$ are deleted.
\end{enumerate}

\begin{lemma}
\label{lem:prefix}
Given an index collection of $N$ numbers, we can compute all prefix sums in the
I/O-memory-bound MapReduce framework in $\OhOf{\log_M N}$ rounds and
$\OhOf{N\log_M N}$ words of communication.
\end{lemma}

\begin{proof}
The fact that the algorithm correctly computes all prefix sums is by
induction on the values $s_{p(v)}$. In each round, each node sends and
receives at most $M$ items, fulfilling the condition of
Theorem~\ref{thm:framework}.
The total number of
rounds is $2L = \OhOf{\log_M N}$ plus the initial round of sending input
elements to the leaves of $\cT$. The total number of items sent in each
round is dominated by items sent by $N$ leaves, which is $\OhOf{N}$ per
round. Applying Theorem~\ref{thm:framework} completes the proof.
\end{proof}

Quite often, the input to the MapReduce computation is a collection of
items with no particular ordering or indexing. If each input element is
annotated with an estimate $N \le \hat N \le N^c$ of the size of
the input, for some constants $c \ge 1$, then using the all-prefix-sum
algorithm we can generate a random indexing for the input with high
probability.

We modify the all-prefix-sum algorithm above as follows. We define the tree
$\cT$ on $\hat N^3$ leaves, thus, the height of the tree is $L =
\ceil{3\log_d \hat N}$. In the initialization step, each input node picks a
random index $i$ in the range $[0, \hat N^3 - 1]$ and sends $a_i = 1$ to
the leaf node $v = (L-1, i)$ of $\cT$. Let $n_v$ be the number of items
that leaf $v$ receives. Note it is possible that $n_v > 1$, thus, we perform
the all-prefix-sums computation with the following differences at the leaf
nodes.  During the bottom-up phase, we define $s_v = n_v$ at the leaf node
$v$.  At the end of the top-down phase, each leaf $v$ assigns each of the
item that it received from the input nodes the indices $s_{p(v)}+1,
s_{p(v)}+2, \ldots, s_{p(v)}+n_v$ in a random order, which is the final
output of the computation.
  
\begin{lemma} \label{lem:index}
 A random indexing of the input can be performed on a collection of data in
 the I/O-memory-bound MapReduce framework in $\OhOf{\log_M N}$ rounds and
 $\OhOf{N \log_M N}$ words of communication with high probability.
\end{lemma}

\begin{proof}
First, note that the probability that $n_v > M$ at some leaf vertex is at
most $N^{-\Omega(M)}$.  Thus, with probability at least $1 -
N^{-\Omega(M)}$, no leaf and, consequently, no node of $\cT$ receives more
than $\OhOf{M}$ elements. Second, note that at most $N$ leaves of the tree
$\cT$ have $A_v(r) \not = \emptyset$.  Since we do not maintain the edges
of the tree explicitly, the total number of items sent in each round is
again dominated by the items sent by at most $N$ leaves, which is
$\OhOf{N}$ per round. Finally, the round and communication complexity
follows from Lemma~\ref{lem:prefix}.
\end{proof}

\section{Simulating BSP and CRCW PRAM Algorithms}
\label{sec:simulation}

In this section we show how to simulate BSP and CRCW PRAM algorithms in the
MapReduce framework.
Our methods therefore provide extensions of the simulation result
of Karloff {\it et al.}~\cite{ksv-amcfm-10}, who show how to optimally
simulate any EREW PRAM algorithm in the MapReduce
framework.\footnote{Their original proof
    was identified for the CREW PRAM model, but there was a flaw in
    that version, which could violate the I/O-buffer-memory size
    constraint during a CREW PRAM simulation. 
    Based on a personal
    communication, we have learned that the subsequent version of
    their paper will identify their proof as being for the EREW PRAM.}

\subsection{Simulating BSP algorithms}
In the BSP model~\cite{v-bmpc-90}, the input of size $N$ is distributed
among $P$ processors so that each processor contains at most $M = \lceil
N/P\rceil$ input items. A computation is specified as a series of
super-steps, each of which involves each processor performing an internal
computation and then sending a set of up to $M$ messages to other
processors.

The initial state of the BSP algorithm is an indexed set of processors
$\{p_1, p_2, \ldots, p_P\}$ and an indexed set of initialized memory cells
$\{m_{1,1}, m_{1,2}, \ldots, m_{p,m}\}$, such that $m_{i,j}$ is the $j$-th
memory cell assigned to processor $i$. Since our framework is almost
equivalent to the BSP model, the simulation is straightforward:
\begin{itemize}
\item Each processor $p_i\ (1 \le i \le P)$ defines a node $v_i$ in our
  generic MapReduce graph $G$, and the internal state $\pi_i$ of $p_i$ and
  its memory cells $\{m_{i,1}, \ldots, m_{i,m}\}$ define the items
  $A_{v_i}$ of node $v_i$.  In the BSP algorithm, in each super-step each
  processor $p_i$ performs a series of computation, updates its internal
  state and memory cells to $\pi'_i$ and $\{m'_{i,1}, \ldots, m'_{i,m}\}$,
  and sends a set of messages $\mu_{j_1}, \ldots, \mu_{j_k}$ to processors
  $p_{j_1}, \ldots, p_{j_k}$, where the total size of all messages sent or
  received by a processor is at most $M$.  In our MapReduce
  simulation, function $f$ at node $v_i$ performs the same computation,
  modifies items $\{\pi_i, m_{i,1}, \ldots, m_{i,m}\}$ to $\{\pi'_i,
  m'_{i,1}, \ldots, m'_{i,m}\}$ and sends items $\mu_{j_1}, \ldots,
  \mu_{j_k}$ to nodes $v_{j_1}, \ldots, v_{j_k}$.
\end{itemize}

\begin{theorem}
\label{thm:bsp}
Given a BSP algorithm $\cal A$ that runs in $R$ super-steps with a total
memory size $N$ using $P \le N$ processors, we can simulate $\cal A$ using
$\OhOf{R}$ rounds and $C = \OhOf{RN}$ communication in the
I/O-memory-bound MapReduce framework with reducer memory size bounded by
$M = \lceil N/P\rceil$.
\end{theorem}

\paragraph{Applications.} 
By Theorem~\ref{thm:bsp}, we can directly simulate BSP algorithms for
sorting~\cite{g-ceps-99} and convex hulls~\cite{g-rfsbt-97},
achieving, for each problem, $\OhOf{\log_M N}$ rounds and $\OhOf{N\log_M N}$
communication complexity.

In Section~\ref{sec:sorting} we will present a randomized sorting
algorithm, which has the same complexity but is simpler than directly
simulating the complicated BSP algorithm in ~\cite{g-ceps-99}.

\subsection{Simulating CRCW PRAM algorithms}
In this section we present a simulation of any $f$-CRCW PRAM model, the
strongest variant of the PRAM model, where concurrent writes to the same
memory location are resolved by applying a commutative semigroup operator
$f$ on all values being written to the same memory address, such as {\em
  Sum}, {\em Min}, {\em Max}, etc.

The input to our simulation of a PRAM algorithm $\cal A$ assumes that the
input is specified by an indexed set of $P$ processor items, $p_1, \ldots,
p_P$, as well as an indexed set of initialized PRAM memory cells, $m_1,
\ldots, m_N$, where $N$ is the total memory size used by $\cal A$.

The main challenge in simulating the algorithm $\cal A$ in the MapReduce
model is that there may be as many as $P$ reads and writes to the same
memory cell in any given step and $P$ can be significantly larger than $M$,
the memory size of reducers. Thus, we need to have a way to ``fan in''
these reads and writes. We accomplish this by using {\em invisible funnel}
technique, where we imagine that there is a different implicit
$\OhOf{M}$-ary tree rooted at each memory cell that has the set of
processors as its leaves.  Intuitively, our simulation algorithm involves
routing reads and writes up and down these $N$ trees. We view them as
``invisible'', because we do not actually maintain them explicitly, since
that would require $\Theta(PN)$ additional memory cells.

The invisible funnels constructed here are similar to the one constructed
for computing random indexing in Section~\ref{sec:prefix}, each of which is
a multi-way tree with fan-out parameter $d= M/2$ and height $L =
\lceil \log_d P \rceil = \OhOf{\log_M P}$. Recall that according to our
labeling scheme, given a node $v = (j, l, k)$, the $k$-{th} node on level
$l$ of the $j$-th tree, we can uniquely identify the label of its parent
$p(v)$ and each of its $d$ children.

We view the computation specified in a single step in the algorithm $\cal
A$ as being composed of a read sub-step, followed by a constant-time
internal computation, followed by a write sub-step. At the initialization
step, we send $m_j$ to the root node of the $j$-{th} tree, i.e., $m_j$ is
sent to node $(j, root) = (j, (0,0))$. For each processor $p_i\ (1 \le i
\le P)$, we send items $p_i$ and $\pi_i$ to node $u_i$, where $\pi_i$
is the internal state of processor $p_i$. Again, throughout the algorithm,
each node keeps the items that it has received in previous rounds until
they are explicitly deleted.

\begin{enumerate}
\item{Bottom-up read phase.} For each processor $p_i$ that attempts to
  read memory location $m_j$, node $u_i$ sends an item encoding a read
  request (in the following we simply say a read request)
  to the $i$-{th} leaf node of the $j$-{th} tree, i.e. to node $(j, L-1,
  i)$, indicating that it would like to read the contents of the $j$-th
  memory cell.

  For $l = L-1$ downto $1$ do:
  \begin{itemize}
  \item For each node $v$ at level $l$, if it received read request(s) in
    the previous round, then it sends a read request to its parent $p(v)$.
  \end{itemize}

\item{Top-down read phase.} The root node in the $j$-{th} tree sends
  the value $m_j$ to child $(j, w_k)$ if child $w_k$ has sent a read
  request at the end of the bottom-up read phase.

  For $l = 1$ to $L-2$ do:
  \begin{itemize}
  \item For each node $v$ at level $l$, if it received $m_j$ from its
    parent in the previous round, then it sends $m_j$ to all those children
    who have sent $v$ read requests during the bottom-up read phase. After
    that $v$ deletes all of its items.
  \end{itemize}
  For each leaf $v$, it sends $m_j$ to the node $u_i\ (1 \le i \le P)$ if
  $u_i$ has sent $v$ a read request at the beginning of the bottom-up read
  phase. After that $v$ deletes all of its items.

\item {Internal computation phase.} At the end of the top-down phase,
  each node $u_i$ receives its requested memory item $m_j$, it performs the
  internal computation, and then sends an item $z$ encoding a write request
  to the node $(j, L-1, i)$ if processor $p_i$ wants to write $z$ to the
  memory cell $m_j$.

\item {Bottom-up write phase.} For $l = L-1$ downto $0$ do:
  \begin{itemize}
    \item For each node $v$ at level $l$, if it received write request(s)
      in the previous round, let $z_1, \ldots, z_k\ (k \le d)$ be the items
      encoding those write requests. If $v$ is not a root, it applies the
      semigroup function on input $z_1, \ldots, z_k$, sends the result $z'$
      to its parent, and then deletes all of its
      items. Otherwise, if $v$ is a root, it modifies its current memory
      item to $z'$. 
  \end{itemize}
\end{enumerate}

When we have completed the bottom-up write phase, we are inductively ready
for simulating the next step in the PRAM algorithm. We have the following.

\begin{theorem}
\label{thm:pram}
Given an algorithm $\cal A$ in the CRCW PRAM model, with write conflicts
resolved according to a commutative semigroup function such that
$\cal A$ runs in $T$ steps using $P$ processors and $N$ memory cells, we
can simulate $\cal A$ in the I/O-memory-bound MapReduce framework in $R =
\OhOf{T\log_M P}$ rounds and with $C = \OhOf{T(N+P)\log_M (N+P)}$
communication complexity.
\end{theorem}
\begin{proof}
Each round in the CRCW PRAM
algorithm is simulated by $\OhOf{\log_M P}$ rounds in the I/O-memory-bound
MapReduce algorithm, and the total number of items sent is $\OhOf{N}$ per round.
\end{proof}

\paragraph{Applications.} 
By Theorem~\ref{thm:pram}, we can directly simulate any CRCW (thus, also
CREW) PRAM algorithm, in particular, linear
programming in fixed dimensions by Alon and Megiddo~\cite{am-plpfd-94}. The
simulation achieves $\OhOf{\log_M N}$ rounds and $\OhOf{N\log_M N}$ communication
complexity.

\section{Multi-searching and Sorting}
\label{sec:search-sorting}
In this section, we present a method for performing simultaneous searches
on a balanced search tree data
structure.  Let $\cT$ be a balanced binary search tree and $Q$ be a set of
queries. Let $N = |\cT| + |Q|$.  The problem of multi-search asks to
annotate each query $q \in Q$ with a leaf $v \in \cT$, such that the
root-to-leaf search path for $q$ in $\cT$ terminates at $v$.

Goodrich~\cite{g-rfsbt-97} provides a solution to the multi-search problem
in the BSP model. However, directly simulating the BSP algorithm in the
I/O-memory-bound MapReduce model has two issues.

First, the model used in~\cite{g-rfsbt-97} is a non-standard BSP model for
it allows a processor to keep an {\em unlimited} number of items between
rounds while still requiring each processor to send and receive at most
$\lceil N/P \rceil = M$ items.
However, a closer inspection of~\cite{g-rfsbt-97} reveals that
the probability that some processor will contain more than $M$ items
in some round is at most $N^{-c}$ for any constant $c \ge 1$.
Therefore, with high probability it can still be simulated in our MapReduce
framework.

Second, the BSP solution requires $\OhOf{N\log_M N}$ space.  Thus,
Theorem~\ref{thm:bsp} provides us with a MapReduce algorithm with
communication complexity $\OhOf{N \log_M^2 N}$.
In this section we improve this communication complexity to $\OhOf{N \log_M
  N}$, while still achieving $\OhOf{\log_M N}$ round complexity with high
probability.

In section~\ref{sec:queue} we also describe a queuing strategy which
reduces the probability of failure due to the first issue of the
simulation. The queuing algorithm might also be of independent interest
because it removes some of the requirements of the framework of
Section~\ref{sec:genericMR}.

\subsection{Multi-searching} \label{sec:search}
As mentioned before, with high probability we can simulate the BSP
algorithm of Goodrich~\cite{g-rfsbt-97} in MapReduce model in $R =
\OhOf{\log_M N}$ rounds and $C = \OhOf{N \log^2_M N}$ communication
complexity.  In this section we present a solution to reduce the
communication complexity by a $\OhOf{\log_M N}$ factor.

The main reason for the large communication complexity of the simulation is
the $\OhOf{N\log_M N}$ size of the search structure that the BSP algorithm
constructs to relieve the congestion caused by multiple queries passing
through the same node of the search tree. It is worth noting that if
the number of queries is small relative to the size of the search tree,
that is, if $|Q| \le N/\log_M N$, then the size of the BSP search structure
is only linear and we can perform the simulation of the algorithm with
$\OhOf{N \log_M N}$ communication complexity.  Thus, for the remainder of
this section we assume $|Q| > N/\log_M N$.

Consider a MapReduce algorithm $\cal A$ that simulates the BSP algorithm
for a smaller set of queries, namely $Q'$ of size only $\ceil{N/\log_{M}
  N}$. Given a search tree $\cT$, algorithm $\cal A$ converts $\cT$ into a
  directed acyclic graph
(DAG) $G$ (see \cite{g-rfsbt-97} for details). $G$ has $\log_M N$ levels and
$\OhOf{N/\log_M N}$ nodes in each level (thus $\OhOf{N/\log_M N}$ source nodes). Therefore the size of $G$ is
$\OhOf{N}$. Next, $\cal A$ propagates the queries of $Q'$ through $G$.  In
each round, with high probability, all queries are routed one level down in
$G$. Thus, the round complexity of $\cal A$ is still $\OhOf{\log_M N}$
while the communication complexity is $\OhOf{N \log_M N}$.

To solve the multi-search problem on the input set $Q$, we make use of
$\cal A$ as follows.  We partition the set of queries $Q$ into $\log_M N$
random subsets $Q_1, Q_2, \dots, Q_{\log_M N}$ each containing
$\OhOf{N/\log_M N}$ queries.  Next, we construct $G$ for the query set
$Q_1$ and also use it to propagate the rest of the query sets. In
particular, we proceed in $\ThetaOf{\log_M N}$ rounds. In each of the first
$\log_M N$ rounds we feed new subset $Q_i$ of queries to the
$\OhOf{N/\log_M N}$ source nodes of $G$ and propagate the queries down to
the sinks using algorithm $\cal A$. This approach can be viewed as a
pipelined execution of $\log_M N$ multi-searches on $G$.

We implement random partitioning of $Q$ by performing a random indexing for
$Q$ (Lemma~\ref{lem:index}) and assigning query with index $j$ to subset
$Q_{\ceil{j/\log_M N}}$. A node $v$ containing a query $q \in Q_i$ keeps
$q$ (by sending it to itself) until round $i$, at which point it sends $q$
to the appropriate source node of $G$.

\begin{theorem}
Given a binary search tree $\cT$ of size $N$, we can perform a multi-search
of $N$ queries over $\cT$ in the I/O-memory-bound MapReduce model in
$\OhOf{\log_M N}$ rounds with $\OhOf{N \log_M N}$ communication with high
probability.
\end{theorem}

\begin{proof} 
We sketch the proof here.  Let $L_1, \ldots, L_{\log_M N}$ be the $\log_M N$
levels of nodes of $G$. First, all query items in the first query batch
$Q_1$ can pass (i.e., be routed down) $L_j\ (1 \le j \le \log_M N)$ in one round with high
probability. This is because for each node $v$ in $L_j$, at most $M$
query items of $Q_1$ will be routed to $v$ with probability at least $1 -
N^{-c}$ for any constant $c$. By taking the union of all the nodes
in $L_j$, we have that with probability at least $1 - \OhOf{N/ \log_M N} \cdot
N^{-c}$, $Q_1$ can pass $L_j$ in one round. Similarly, we can prove
that any $Q_i\ (1 \le i \le \log_M N)$ can pass $L_j\ (1 \le j \le \log_M
N)$ in one round with the same probability since sets $Q_i$ have equal
distributions. Since there are $\log_M N$ batches of queries and they are
fed into $G$ in a pipeline fashion, by union bound we have that with
probability at least $1 - \log_M^2 N \cdot \OhOf{N/\log_M N} \cdot N^{-c} \ge 1 -
1/N$ (by choosing a sufficient large constant $c$) the whole process
completes within $\OhOf{\log_M N}$ rounds. The communication complexity
follows directly since we only need to send $\OhOf{\abs{G} + \abs{Q}} =
\OhOf{N}$ items in each round.
\end{proof}

\subsection{FIFO Queues in MapReduce Model}\label{sec:queue}
As mentioned at the
beginning of this section, with probability $1-N^{-c}$ for any constant $c \ge
1$ no processor in the BSP algorithm for multi-searching contains more than $M$
items. Thus, the algorithm for multi-search in the previous section can be
implemented in the I/O-memory-bound MapReduce framework with high
probability.

However, the failure of the algorithm implies a crash of a reducer in the
MapReduce framework, which is quite undesirable. In this section we present a
queuing strategy which ensures that no reducer receives more than $M$ items,
	which might be of independent interest.
  
Consider the following modified version of the generic MapReduce framework
from Section~\ref{sec:genericMR}.  In this version we still require each node
$v \in V$ to send at most $M$ items. However, instead of limiting
the number of items that a node keeps or receives to be $M$, we only
require that in every round at most $M$ different nodes send to any
given node $v$, and function $f$ takes as input a list of at most $M$
items. To accommodate the latter requirement, if a node receives or
contains more than $M$ items, the excess items are kept within the node's
input buffer and are fed into function $f$ in batches of $\OhOf{M}$ items per round
in a first-in-first-out (FIFO) order.

In this section we show that any algorithm $\cal A$ with round complexity
$R$ and communication complexity $C$ in the modified framework can be
implemented using the framework in Section~\ref{sec:genericMR} with the same asymptotic round and communication complexities.
 
We simulate algorithm $\cal A$ by implementing the FIFO queue at each node
$v$ by a doubly-linked list $L_v$ of nodes, such that $L_v \cap L_w =
\emptyset$ for all $v \neq w$ and $L_v \cap V = \emptyset$ for all $v \in
V$. Each node $v \in V$ keeps a pointer
$head_{L_v}$ to the head of its list $L_v$. In addition, $v$ also keeps
$n_{head}$, the number of query items at $head_{L_v}$. If $L_v$ is empty,
$head_{L_v}$ points at $v$ and $n_{head} = 0$. Throughout the algorithm we
maintain an invariant that for each doubly-linked list $L_v$, each node in
$L_v$ contains $[M/4, M/2]$ query items except the head node, i.e., the one
containing the last items to be processed in the queue, which contains at
most $M/2$ query items.  We simulate one round of $\cal A$ by the
following three rounds. Let $\IN(v)$ and $\OUT(v)$ denote the set of in- and
out-neighbors of node $v \in V$, respectively. That is, for each $u \in
\IN(v)$, $(u,v) \in E$ and for each $w \in \OUT(v), (v,w) \in E$.

\begin{itemize}
\item[R1.] Each node $u \in V$ that wants to send $n_{u,v}$ query
  items to $v \in \OUT(u)$, instead of sending the actual query items, 
  sends $n_{u,v}$ to $v$.

\item[R2.] Each node $v \in V$ receives a set of different values $n_{u_1,
  v}, n_{u_2, v}, \dots, n_{u_k, v}$ from its in-neighbors $u_1, u_2,
  \dots, u_k\ (k \le M)$.  For convenience we define $n_{u_0,v}
  \triangleq n_{head}$.  Next, $v$ partitions the set $\{0, 1, \dots, k\}$
  into sets $S_1, \dots, S_{m}$, $m \le k$, such that $M/4 \le \sum_{j \in
    S_i} n_{u_j, v} \le M/2$ for all $1 \le i \le m-1$ and $\sum_{j \in
    S_{m}} n_{u_j, v} \le M/2$.  W.l.o.g., assume that $0 \in S_1$.  For each
  $S_i$, we will have a corresponding node $w_i$ in the list $L_v$: We let
  $w_1 = head_{L_v}$ and for each $S_i$, $2 < i \le m$ we pick a new node
  $w_i$, create edges $(w_i, w_{i-1})$ and $(w_{i-1}, w_i)$, and send it to
  nodes $w_i$ and $ w_{i-1}$, respectively. For each $j \in S_i$, we also
  notify $u_j$ that it should send all its queries to $w_i$ by sending the
  label of $w_i$ to $u_j$.  The only exception to this rule is that if $w_1
  \not = v$ and $w_1$ contains the edge $(w_1, v)$, i.e. it is the first
  node in $L_v$. In this case, for each $j \in S_1$ each $u_j$ should send
  queries directly to $v$. Finally, we update the pointer $head_{L_v}$ to
  point to $w_m$ and update $n_{head} = \sum_{j\in S_m} n_{u_j,v}$, unless
  $w_m = v$, in which case $n_{head} = 0$.

\item[R3.] Each node $u_j \in \IN(v)$ receives the label of a node $w_i$
  from $v$ in the previous rounds. It sends all its query items to
  $w_i$. Note that if $w_i = v$, all items will be sent to $v$ directly.
  At the same time, each node $w \not \in V$, i.e. $w \in L_v$, that has an
  edge $(w,v)$ for some $v \in V$ sends all its items to $v$ and extracts
  itself from the list. The node $w$ accomplishes this by deleting all
  edges incident to $w$ and by sending to its predecessor $\pred{w}$ in the
  queue $L_v$ a new edge $(\pred{w}, v)$, thus, linking the rest of the
  queue to $v$.
\end{itemize}

\begin{theorem} \label{lem:queue}
Consider a modified MapReduce framework, where in every round each node is
required to send at most $M$ items, but is allowed to keep and receive an
unlimited number of items as long as they arrive from at most $M$ different
nodes, with excess items stored in FIFO input buffer and fed into function
$f$ in blocks of size at most $M$. Let $\cal A$ be an algorithm in this
modified MapReduce framework with $R$ round complexity and $C$
communication complexity.  Then we can implement $\cal A$ in the original
I/O-memory-bound MapReduce framework in $\OhOf{R}$ rounds and $\OhOf{C}$
communication complexity.
\end{theorem}

\begin{proof}
First, it is easy to see that our simulation ensures that each node keeps as well as
sends and receives at most $M$ items. Next, note that in every three rounds (round $3t, 3t+1, 3t+2$), each node
$v \in V$ routes $\min\{\ThetaOf{M}, k^t_v\}$ items, where $k^t_v$ is the
combined number of items in the queue $L_v$ and the number of items that
$v$'s in-neighbors send to $v$ during the three rounds. This is within a
constant factor of the number of items that $v$ routes in round $t$ in
algorithm $\cal A$. Finally, the only additional items we send in each round are the edges of
the queues $\{L_v\ |\ v \in V\}$. Note that we only need to maintain
$\OhOf{1}$ additional edges for each node of each $L_v$. And since these
nodes are non-empty, the additional edges do not contribute more than a
constant factor to the communication complexity.
\end{proof}

\vspace{-0.4cm}
\paragraph{Applications.} The DAG $G$ of the multi-search BSP algorithm~\cite{g-rfsbt-97} satisfies the requirement that at most $M$ nodes attempt
to send items to any other node. In addition, if some processor of the BSP
algorithm happens to keep more than $M$ items, the processing of these
items is delayed and can be processed in any order, including FIFO. Thus,
the requirements of Theorem~\ref{lem:queue} are satisfied.

We do not know how to modify our random indexing algorithm
in Section~\ref{sec:prefix} to fit the modified framework. Thus, we cannot
provide a Las Vegas algorithm. However, the above framework reduces the
probability of failure from $N^{-\Omega(1)}$ to the probability of failure
of the random indexing step, i.e., $N^{-\Omega(M)}$, which is much smaller
for large values of $M$.

The modified framework might be of independent interest because it allows
for an alternative way of designing algorithms for MapReduce. In
particular, it removes the burden of keeping track of the number of items
kept or sent by a node.

\subsection{Sorting}
\label{sec:sorting}
In this section, we show how to obtain a simple sorting algorithm in the
MapReduce model by using our multi-search algorithm. First, it is easy to
obtain the following brute-force sorting result, which is proved in
Appendix~\ref{sec:bruteforce}.
\begin{lemma}
\label{lem:brute-sort}
Given a set $X$ of $N$ indexed comparable items, we can sort them in
$\OhOf{\log_M N}$ rounds and $\OhOf{N^2 \log_M N}$ communication complexity
in the MapReduce model.
\end{lemma}
Combining the brute-force sorting algorithm with the multi-searching
algorithm in the previous section, we present here a simple sorting
algorithm with optimal round and communication complexities.

\begin{enumerate}
  \item Pick $\Theta(\sqrt{N})$ random pivots. Sort the pivots using
    brute-force sorting algorithm. This results in the pivots being
    assigned a unique index/label in the range $[1, \sqrt{N}]$.
  \item  Build a search tree on the set of pivots as the leaves of the tree. 
  \item Perform a multi-search on the input items over the search tree. The
    result is the label associated with each item which is equal to the
    ``bucket'' within which the input is partitioned into.
  \item  Recursively sort each bucket in parallel.
\end{enumerate}
Combined with Lemma~\ref{lem:brute-sort} it is easy to see that this sorting
algorithm runs in $\OhOf{\log_M N}$ rounds and has $\OhOf{N\log_M N}$
communication complexity with high probability.

{\raggedright \bibliographystyle{abbrv} \bibliography{par,cuckoo,paper} }

\begin{thebibliography}{10}

\bibitem{av-iocsr-88}
A.~Aggarwal and J.~S. Vitter.
\newblock The input/output complexity of sorting and related problems.
\newblock {\em Commun. ACM}, 31:1116--1127, 1988.

\bibitem{am-plpfd-94}
N.~Alon and N.~Megiddo.
\newblock Parallel linear programming in fixed dimension almost surely in
  constant time.
\newblock {\em J. ACM}, 41(2):422--434, 1994.

\bibitem{a-emaa-97}
L.~Arge.
\newblock External-memory algorithms with applications in gis.
\newblock In {\em Algorithmic Foundations of Geographic Information Systems},
  pages 213--254, London, UK, 1997. Springer-Verlag.

\bibitem{a-emds-02}
L.~Arge.
\newblock External memory data structures.
\newblock In {\em Handbook of massive data sets}, pages 313--357. Kluwer
  Academic Publishers, Norwell, MA, USA, 2002.

\bibitem{dg-msdpl-08}
J.~Dean and S.~Ghemawat.
\newblock {MapReduce}: simplified data processing on large clusters.
\newblock {\em Commun. ACM}, 51(1):107--113, 2008.

\bibitem{dg-mafdp-10}
J.~Dean and S.~Ghemawat.
\newblock {MapReduce}: a flexible data processing tool.
\newblock {\em Commun. ACM}, 53(1):72--77, 2010.

\bibitem{ds-mmsb-08}
D.~J. DeWitt and M.~Stonebraker.
\newblock {MapReduce}: A major step backwards.
\newblock {\em Database Column}, 2008.
\newblock
  http://databasecolumn.vertica.com/database-innovation/mapreduce-a-major-step%
-backwards/.

\bibitem{e-oascp-07}
S.~C. Eisenstat.
\newblock {$O(\log^* n)$} algorithms on a {Sum-CRCW PRAM}.
\newblock {\em Computing}, 79(1):93--97, 2007.

\bibitem{DBLP:conf/soda/FeldmanMSSS08}
J.~Feldman, S.~Muthukrishnan, A.~Sidiropoulos, C.~Stein, and Z.~Svitkina.
\newblock On distributing symmetric streaming computations.
\newblock In S.-H. Teng, editor, {\em SODA}, pages 710--719. SIAM, 2008.

\bibitem{g-rfsbt-97}
M.~T. Goodrich.
\newblock Randomized fully-scalable {BSP} techniques for multi-searching and
  convex hull construction.
\newblock In {\em Proc. ACM-SIAM Sympos. Discrete Algorithms (SODA)}, pages
  767--776, 1997.

\bibitem{g-ceps-99}
M.~T. Goodrich.
\newblock Communication-efficient parallel sorting.
\newblock {\em SIAM Journal on Computing}, 29(2):416 -- 432, 1999.

\bibitem{j-ipa-92}
J.~J\'aJ\'a.
\newblock {\em An Introduction to Parallel Algorithms}.
\newblock Addison-Wesley, Reading, Mass., 1992.

\bibitem{kannan:implicitgraphs}
S.~Kannan, M.~Naor, and S.~Rudich.
\newblock Implicit representation of graphs.
\newblock In {\em 20th Annual ACM Symposium on Theory of Computing (STOC)},
  pages 334--343, 1988.

\bibitem{ksv-amcfm-10}
H.~Karloff, S.~Suri, and S.~Vassilvitskii.
\newblock A model of computation for {MapReduce}.
\newblock In {\em Proc. ACM-SIAM Sympos. Discrete Algorithms (SODA)}, pages
  938--948, 2010.

\bibitem{kf-capda-67}
H.~Kucera and W.~N. Francis.
\newblock {\em Computational Analysis of Present-Day American English}.
\newblock Brown University Press, Providence, RI, 1967.

\bibitem{rrbn-smlsa-09}
M.~M. Rafique, B.~Rose, A.~R. Butt, and D.~S. Nikolopoulos.
\newblock Supporting {MapReduce} on large-scale asymmetric multi-core clusters.
\newblock {\em SIGOPS Oper. Syst. Rev.}, 43(2):25--34, 2009.

\bibitem{v-bmpc-90}
L.~G. Valiant.
\newblock A bridging model for parallel computation.
\newblock {\em Comm.\ ACM}, 33:103--111, 1990.

\bibitem{v-emads-01}
J.~S. Vitter.
\newblock External memory algorithms and data structures: dealing with massive
  data.
\newblock {\em ACM Comput. Surv.}, 33(2):209--271, 2001.

\bibitem{DBLP:reference/algo/Vitter08}
J.~S. Vitter.
\newblock External sorting and permuting.
\newblock In M.-Y. Kao, editor, {\em Encyclopedia of Algorithms}. Springer,
  2008.

\bibitem{w-hdg-09}
T.~White.
\newblock {\em Hadoop: The Definitive Guide}.
\newblock O'Reilly Media, Inc., 2009.

\end{thebibliography}
\clearpage

\begin{appendix}

\section{Brute-Force Multi-search and Sorting}
\label{sec:bruteforce}
In this section, we first present a brute force multi-search algorithm, and
then use it to design a brute force sorting algorithm.

The input for multi-search is a set $X = \{x_1, \ldots, x_n\}$ of $n$ query
items and a sorted set $Y = \{y_1, \ldots, y_n\}$ of $m$ items
corresponding the the leaves of the search tree. The goal is for each
$x_i\ (1 \le i \le n)$ to find the leaf $y_j$ such that search path for $x_i$
will terminate at $y_j$. Moreover, for each leaf $y_j\ (j = 1, \ldots,
m)$, we want to compute the number of items in $X$ whose search paths will terminate
at $y_j$. The input for sorting is a set $X = \{x_1, \ldots, x_n\}$ of $n$
items. The goal is to sort the $n$ items. We assume that set $X$ is
indexed, otherwise we can first perform the random indexing by
Lemma~\ref{lem:index}. Note that the set $Y$ is sorted thus indexed by
default.

\paragraph{Brute-Force Multi-search}
At the beginning of the algorithm, let nodes $\{p_1, \ldots, p_n\}$ be the
input nodes containing input items $x_1, \ldots, x_n$, respectively.  And
nodes $\{q_1, \ldots, q_n \}$ be the input nodes containing input items
$y_1, \ldots, y_n$, respectively. The input items are always kept during
the computation.
\begin{enumerate}
\item{Generate all pairs:} First, $p_i$ sends $x_i$ to node $v_{i,1}$ for
  all $i = 1, \ldots, n$, and $q_j$ sends $y_j$ to node $v_{1,j}$ for
  all $j = 1, \ldots, m$. Next, for $l = 1$ to $\log_M m$ do:
  \begin{itemize}
  \item For each $i \in [n]$, for each node $v_{i,j}$ containing an item
    $x_i$, it keeps $x_i$ and sends a copy of $x_i$ to nodes $v_{i,j'_1},
    \ldots, v_{i,j'_M}$, where $j'_k = (j-1) \cdot M + k$ for $1 \le k \le
    M$.
  \end{itemize}
  Similarly, for $l = 1$ to $\log_M n$ do:
  \begin{itemize}
  \item For each $j \in [m]$, for each node $v_{i,j}$ containing an item
    $y_j$, it keeps $y_j$ and sends a copy of $y_j$ to nodes $v_{i'_1,j},
    \ldots, v_{i'_M,j}$, where $i'_k = (i-1) \cdot M + k$ for $1 \le k \le
    M$.
  \end{itemize}

\item{Compare each pair of items:} Each node $v_{i,j}$ compares its item
  $x_i$ and $y_j$. If $x_i \le y_j$, $v_{i,j}$ generates $(x_i, 0)$ and
  keeps it; otherwise it generates $(x_i, 1)$ and keeps it.

\item{Add up values:} For each $i \in [n]$, let $(x_i, b_j)\ (b_j \in
  \{0,1\})$ be the item stored at node $v_{i,j}$ for $j = 1,\ldots,m$. We
  compute $k_i = \sum_{j=1}^m b_j$ in the same way as the bottom-up phase
  of computing the prefix sums in Section~\ref{sec:prefix}. Then $y_{k_i}$
  is the leaf node in $Y$ where the search path of $x_i$ ends.

  Similarly, for each $j \in [m]$, let $(x_i, b_j)\ (b_j \in \{0,1\})$ be
  the item stored at node $v_{i,j}$ for $i = 1,\ldots,n$. We compute $c_i =
  \sum_{i=1}^n b_j$, which is the number of query items in $X$ whose search
  paths end at $y_j$.
\end{enumerate}

\paragraph{Brute-Force Sorting}
The brute force sorting can be solved by the brute force multi-search
algorithm. We create a copy of $X$ and think it as the set $Y$. And then we
run the algorithm for multi-search. The $k_i$ computed for each $x_i$ is
the rank of item $x_i$.
\end{appendix}

\end{document}